\documentclass{article}

\usepackage{arxiv}

\usepackage[utf8]{inputenc} % allow utf-8 input
\usepackage[T1]{fontenc}    % use 8-bit T1 fonts
\usepackage{hyperref}       % hyperlinks
\usepackage{url}            % simple URL typesetting
\usepackage{booktabs}       % professional-quality tables
\usepackage{amsfonts}       % blackboard math symbols
\usepackage{nicefrac}       % compact symbols for 1/2, etc.
\usepackage{microtype}      % microtypography
\usepackage{lipsum}

\usepackage{graphicx}
\usepackage{amsthm}
\usepackage{amssymb}
\usepackage{amsmath}
\newtheorem{definition}{Definition}
\newtheorem{theorem}{Theorem}
\newtheorem{example}{Example}

\begin{document}

\title{Mechanism Design for Public Projects via Neural Networks}
%\titlenote{Produces the permission block, and copyright information}

% AAMAS: as appropriate, uncomment one subtitle line; check the CFP
%\subtitle{Extended Abstract}
%\subtitle{Blue Sky Ideas Track}
%\subtitle{JAAMAS Track}
%\subtitle{Demonstration}
%\subtitle{Doctoral Consortium}

%\author{Paper \#1840}
%\author{Guanhua Wang, Runqi Guo, Yuko Sakurai, Ali Babar, Mingyu Guo}

\author{
  Guanhua Wang, Runqi Guo\\
  School of Computer Science\\
  University of Adelaide\\
  Australia\\
   \And
   Yuko Sakurai\\
   National Institute of Advanced\\
   Industrial Science and Technology\\
   Japan\\
   \And
  Ali Babar, Mingyu Guo\\
  School of Computer Science\\
  University of Adelaide\\
  Australia\\
}

\maketitle

\begin{abstract} We study mechanism design for nonexcludable and excludable
    binary public project problems.  We aim to maximize the expected number of
    consumers and the expected social welfare.  For the nonexcludable public
    project model, we identify a sufficient condition on the prior distribution
    for the conservative equal costs mechanism to be
    the optimal strategy-proof and individually rational mechanism. For general
    distributions, we propose a dynamic program that solves for the optimal
    mechanism.  For the excludable public project model, we identify a similar
    sufficient condition for the serial cost sharing
    mechanism to be optimal for $2$ and $3$ agents. We
    derive a numerical upper bound. Experiments show that for several common
    distributions, the serial cost sharing mechanism is close to optimality.

The serial cost sharing mechanism is not optimal in general.  We design better
    performing mechanisms via neural networks. Our approach involves several
    technical innovations that can be applied to mechanism design in general.
    We interpret the mechanisms as price-oriented rationing-free
    (PORF) mechanisms, which enables us to
    move the mechanism's complex (\emph{e.g.}, iterative) decision making off
    the network, to a separate program.  We feed the prior distribution's
    analytical form into the cost function to provide quality gradients for
    training.  We use supervision to manual mechanisms as a systematic way for
    initialization.  Our approach of ``supervision and then gradient descent''
    is effective for improving manual mechanisms' performances. It is also
effective for fixing constraint violations for heuristic-based mechanisms that
are infeasible.  \end{abstract}


\section{Introduction}\label{sec:intro} Many multiagent system applications (\emph{e.g.},
crowdfunding) are related to the public project problem. The public project
problem is a classic economic model that has been studied extensively in both
economics and computer
science~\cite{Mas-Colell1995:Microeconomic,Moore2006:General,Moulin1988:Axioms}.
Under this model, a group of agents decide whether or not to fund a
\emph{nonrivalrous} public project --- when one agent consumes the project, it
does not prevent others from using it.

We study both the \textbf{nonexcludable} and the \textbf{excludable} versions
of the \emph{binary} public project problem. The binary decision is either to
build or not. If the decision is not to build, then no agents can consume the
project.  For the \emph{nonexcludable} version, once a project is built, all
agents can consume it, including those who do not pay.  For example, if the
public project is an open source software project, then once the project is
built, everyone can consume it.  For the \emph{excludable} version, the
mechanism has the capability to exclude agents from the built project. For
example, if the public project is a swimming pool, then we could impose the
restriction that only some agents (\emph{e.g.}, the paying agents) have access
to it.

Our aim is to design mechanisms that maximize \emph{expected} performances.  We
consider two design objectives. One is to maximize the \textbf{expected number
of consumers} (expected number of agents who are allowed to consume the
project).\footnote{For the nonexcludable public project model, this is simply
to maximize the probability of building, as the number of consumers is always
the total number of agents if the project is built.  } The other objective is
to maximize the agents' \textbf{expected social welfare} (considering payments). It should be noted
that for some settings, we obtain the same optimal mechanism under these two
different objectives. In general, the optimal mechanisms differ.

We argue that maximizing the expected number of consumers is \emph{more fair}
in some application scenarios.  When maximizing the social welfare, the main
focus is to ensure the high-valuation agents are served by the project, while
low-valuation agents have much lower priorities. On the other hand, if the
objective is to maximize the expected number of consumers, then low-valuation
agents are as important as high-valuation agents.

Guo~\emph{et.al.}~\cite{Guo2018:Cost} studied an objective that is very similar
to maximizing the expected number of consumers.  The authors studied the
problem of crowdfunding security information.  There is a premium time period.
If an agent pays more, then she receives the information earlier. If an agent
pays less or does not pay, then she incurs a time penalty --- she receives the
information slightly delayed.  The authors' objective is to minimize the
expected delay.  If every agent either receives the information at the very
beginning of the premium period, or at the very end, then minimizing the
expected delay is equivalent to maximizing the expected number of consumers.
The public project is essentially the premium period. It should be noted that
when crowdfunding security information, it is desirable to have more agents
protected, whether their valuations are high or low. Hence, in this application
domain, maximizing the number of consumers is more suitable than maximizing
social welfare.  However, since any delay that falls \emph{strictly} inside the
premium period is not valid for our \emph{binary} public project model, the
mechanisms proposed in~\cite{Guo2018:Cost} do not apply to our setting.

% \vspace{.1in} We adopt the characterization results from
% Ohseto~\cite{Ohseto2000:Characterizations} for \emph{strategy-proof} and
% \emph{individually rational} mechanisms for both the nonexcludable and the
% excludable public project models.  For the nonexcludable version, we only
% need to focus on the family of \emph{unanimous mechanisms}.  An example
% unanimous mechanism is the \emph{conservative equal cost
% mechanism}~\cite{Moulin1994:Serial}.  For the excludable version, we only
% need to consider the family of \emph{largest unanimous mechanisms}. For the
% excludable version, the characterizations require two additional technical
% conditions called \emph{Demand Monotonicity} and \emph{Access Independence}.
% An example largest unanimous mechanism is the \emph{serial cost sharing
% mechanism}~\cite{Moulin1994:Serial}.
With slight technical adjustments, we adopt the existing characterization
results from Ohseto~\cite{Ohseto2000:Characterizations} for
\emph{strategy-proof} and \emph{individually rational} mechanisms for both the
nonexcludable and the excludable public project problems.  Before summarizing
our results, we introduce the following notation. We assume the agents'
valuations are drawn independently and identically from a known distribution,
with $f$ being the probability density function.

For the nonexcludable public project problem, we propose a sufficient condition
for the \emph{conservative equal costs mechanism}~\cite{Moulin1994:Serial} to
be optimal.  For maximizing the expected number of consumers, $f$ being
\emph{log-concave} is a sufficient condition. For maximizing social welfare, besides
log-concavity, we propose a condition on $f$ called
\emph{welfare-concavity}.  For distributions not satisfying the above
conditions, we propose a dynamic program that solves for the optimal mechanism.

For the excludable public project problem, we also propose a sufficient
condition for the \emph{serial cost sharing mechanism}~\cite{Moulin1994:Serial}
to be optimal.  Our condition only applies to cases with $2$ and $3$ agents.
For $2$ agents, the condition is identical to the nonexcludable version.  For
$3$ agents, we also need $f$ to be nonincreasing.  For more agents, we propose
a numerical technique for calculating the objective upper bounds. For a few
example log-concave distributions, including common distributions like uniform
and normal, our experiments show that the serial cost sharing mechanism is
close to optimality.

Without log-concavity, the serial cost sharing mechanism can be far away from
optimality. We propose a neural network based approach, which successfully
identifies better performing mechanisms.  Mechanism design via deep
learning/neural networks has been an emerging topic~\cite{Golowich2018:Deep,
Duetting2019:Optimal,Shen2019:Automated,Manisha2018:Learning}.  Duetting
\emph{et.al.}~\cite{Duetting2019:Optimal} proposed a general approach for
revenue maximization via deep learning. The high-level idea is to manually
construct often complex network structures for representing mechanisms for
different auction types. The cost function is the negate of the revenue. By
minimizing the cost function via gradient descent, the network parameters are
adjusted, which lead to better performing mechanisms.  The mechanism design
constraints (such as strategy-proofness) are enforced by adding a penalty term
to the cost function. The penalty is calculated by sampling the type profiles
and adding together the constraint violations.  Due to this setup, the final
mechanism is only approximately strategy-proof. The authors demonstrated that
this technique scales better than the classic mixed integer programming based
automated mechanism design approach~\cite{Conitzer2002:Complexity}.  Shen
\emph{et.al.}~\cite{Shen2019:Automated} proposed another neural network based
mechanism design technique, involving a seller's network and a buyer's network.
The seller's network provides a menu of options to the buyers.  The buyer's
network picks the utility-maximizing menu option. An exponential-sized
hard-coded buyer's network is used (\emph{e.g.}, for every discretized type
profile, the utility-maximizing option is pre-calculated and stored in the
network).  The authors mostly focused on settings with only one buyer.

Our approach is different from previous approaches, and it involves three
technical innovations, which have the potential to be applied to mechanism
design in general.

\vspace{.1in}\noindent \emph{Calculating mechanism decisions off the network by
interpreting mechanisms as price-oriented rationing-free (PORF)
mechanisms~\cite{Yokoo2003:Characterization}:} A mechanism often involves
binary decisions (\emph{e.g.}, for an agent, depending on whether her valuation
is above the price offered to her, we end up with different situations). A
common way to model binary decisions on neural networks is by using the
\emph{sigmoid} function (or similar activation functions).  A mechanism may
involve a complex decision process, which makes it difficult or impractical to
model via \emph{static} neural networks.  For example, for our setting, a
mechanism involves \emph{iterative} decision making. We could stack multiple
sigmoid functions to model this.  However, stacking sigmoid functions leads to
vanishing gradients and significant numerical errors. Instead, we rely on the
PORF interpretation: every agent faces a set of options (outcomes with prices)
determined by the other agents. We single out a randomly chosen agent $i$, and
draw a sample of \emph{the other agents' types $v_{-i}$}.  We use a separate program (off the
network) to calculate the options $i$ would face. For example, the separate
program can be any Python function, so it is trivial to handle complex and
iterative decision making. We no longer need to construct complex network
structures like the approach in~\cite{Duetting2019:Optimal} or resort to
exponential-sized hard-coded buyer networks like the approach
in~\cite{Shen2019:Automated}.  After calculating $i$'s options, we link the
options together using terms that carry gradients.  One effective way to do
this is by making use of the prior distribution as discussed below.

\vspace{.1in}\noindent \emph{Feeding prior distribution into the cost
function:} In conventional machine learning, we have access to a finite set of
samples, and the process of machine learning is essentially to infer the true
probability distribution of the samples. For existing neural network mechanism
design
approaches~\cite{Duetting2019:Optimal,Shen2019:Automated}
(as well as this paper), it is assumed that the prior distribution is known.
After calculating agent $i$'s options, we make use of $i$'s distribution to
figure out the probabilities of all the options, and then derive the expected
objective value from $i$'s perspective. We assume that the prior distribution is continuous. If we have the \emph{analytical form}
of the prior distribution, then the probabilities can
provide quality gradients for our training process. This is due to the fact that probabilities are
calculated based on neural network outputs. In summary, we
combine both samples and distribution in our cost function.  We also have an
example showing that even if the distribution we provide is not $100\%$
accurate, it is still useful.  (Sometimes, we do not have the analytical form
of the distribution.  We can then use an analytical approximation instead.)

\vspace{.1in}\noindent
\emph{Supervision to manual mechanisms as initialization:} We start our
training by first conducting supervised learning. We teach the network to mimic
an existing manual mechanism, and then leave it to gradient descent. This is
essentially a systematic way to improve manual mechanisms.\footnote{Of course,
if the manual mechanism is already optimal, or is ``locally optimal'', then the
gradient descent process may fail to find improvement.} In our experiments,
besides the \emph{serial cost sharing mechanism}, we also considered two
heuristic-based manual mechanisms as starting points. One heuristic is feasible
but not optimal, and the gradient descent process is able to improve its
performance. The second heuristic is not always feasible, and the gradient
descent process is able to fix the constraint violations. Supervision to manual
mechanisms is often better than random initializations.  For one thing, the
supervision step often pushes the performance to a state that is already
somewhat close to optimality.  It may take a long time for random
initializations to catch up. In computational expensive scenarios, it may never
catch up.  Secondly, supervision to a manual mechanism is a systematic way to
set good initialization point, instead of trials and errors.  It should be noted
that for many conventional deep learning application domains, such as computer
vision, well-performing manual algorithms do not exist. Fortunately, for
mechanism design, we often have simple and well-performing mechanisms to be
used as starting points.

\section{Model Description}

$n$ agents need to decide whether or not to build a public project.  The
project is \emph{binary} (build or not build) and \emph{nonrivalrous} (the cost
of the project does not depend on how many agents are consuming it).  We
normalize the project cost to $1$.  Agent $i$'s type $v_i\in[0,1]$ represents
her private valuation for the public project. We assume that the $v_i$ are
drawn \emph{i.i.d.} from a known prior distribution. Let $F$ and $f$ be the CDF
and PDF, respectively. We assume that the distribution is continuous and $f$ is
differentiable.

\begin{itemize}

    \item For the nonexcludable public project model, agent $i$'s valuation is
        $v_i$ if the project is built, and $0$ otherwise.

    \item For the excludable public project model, the outcome space is
        $\{0,1\}^n$.  Under outcome $(a_1,a_2,\ldots,a_n)$, agent $i$ consumes
        the public project if and only if $a_i=1$. If for all $i$, $a_i=0$,
        then the project is not built.  As long as $a_i=1$ for some $i$, the
        project is built.

\end{itemize}

We use $p_i\ge 0$ to denote agent $i$'s payment. We require that $p_i=0$ for
all $i$ if the project is not built and $\sum p_i=1$ if the project is built.
An agent's payment is also referred to as her \emph{cost share} of the project.
An agent's utility is $v_i-p_i$ if she gets to consume the project, and $0$
otherwise.

We focus on \emph{strategy-proof} and \emph{individually rational} mechanisms.
We study two objectives. One is to maximize the expected number of consumers.
The other is to maximize the social welfare.

% Under a strategy-proof mechanism, it is a dominant strategy for an agent to
% truthfully report her valuation. Under an individually rational mechanism, an
% agent's utility is at least $0$.

% Let $v=(v_1,v_2,\ldots,v_n)$ be a type profile.  Let $M$ be a mechanism. We
% define the \emph{benefit} function $B_M(v)$ as follows:

% \begin{itemize}

%     \item If the objective is to maximize the expected number of consumers,
%         then $B_M(v)$ denotes the number of consumers under $M$ for type
%         profile $v$.

%     \item If the objective is to maximize the social welfare, then $B_M(v)$
%         denotes the agents' total utility under $M$ for type profile $v$.

% \end{itemize}

% Our goal is to maximize the expectation $E_{v\sim f}(B_M(v))$.

\section{Characterizations and Bounds}

We adopt a list of existing characterization results
from~\cite{Ohseto2000:Characterizations}, which characterizes strategy-proof
and individual rational mechanisms for both nonexcludable and excludable public
project problems.  A few technical adjustments are needed for the existing
characterizations to be valid for our problem.  The characterizations
in~\cite{Ohseto2000:Characterizations} were not proved for quasi-linear
settings. However, we verify that the assumptions needed by the proofs are
valid for our model setting. One exception is that the characterizations
in~\cite{Ohseto2000:Characterizations} assume that every agent's valuation is
strictly positive.  This does not cause issues for our objectives as we are
maximizing for expected performances and we are dealing with continuous
distributions.\footnote{Let $M$ be the optimal mechanism. If we restrict the
valuation space to $[\epsilon,1]$, then $M$ is Pareto dominated by an
unanimous/largest unanimous mechanism $M'$ for the nonexcludable/excludable
setting. The expected performance difference between $M$ and $M'$ vanishes as
$\epsilon$ approaches $0$. Unanimous/largest unanimous mechanisms are still
strategy-proof and individually rational when $\epsilon$ is set to exactly
$0$.} We are also safe to drop the \emph{citizen sovereign} assumption
mentioned in one of the characterizations\footnote{If a mechanism always
builds, then it is not individually rational in our setting.  If a mechanism
always does not build, then it is not optimal.}, but not the other two minor
technical assumptions called \emph{demand monotonicity} and \emph{access
independence}.

\subsection{Nonexcludable Mech. Characterization}

\begin{definition}[Unanimous mechanism~\cite{Ohseto2000:Characterizations}]
    There is a constant cost share vector $(c_1,c_2,\ldots,c_n)$ with $c_i\ge
    0$ and $\sum c_i=1$. The mechanism builds if and only if $v_i\ge c_i$ for
    all $i$. Agent $i$ pays exactly $c_i$ if the decision is to build.  The
unanimous mechanism is strategy-proof and individually rational.
\end{definition}

\begin{theorem}[Nonexcludable mech. characterization~\cite{Ohseto2000:Characterizations}]
    For the nonexcludable public project model,
if a mechanism is strategy-proof, individually rational, and citizen
sovereign, then it is weakly Pareto dominated by an unanimous mechanism.

\noindent Citizen sovereign: Build and not build are both possible outcomes.
\end{theorem}

Mechanism $1$ weakly Pareto dominates Mechanism $2$ if every agent
weakly prefers Mechanism $1$ under every type profile.

\begin{example}[Conservative equal costs mechanism~\cite{Moulin1994:Serial}] An
example unanimous mechanism works as follows: we build the project if and only
if every agent agrees to pay $\frac{1}{n}$.  \end{example}

\subsection{Excludable Mech. Characterization}

\begin{definition}[Largest unanimous
    mechanism~\cite{Ohseto2000:Characterizations}] For every nonempty coalition
    of agents $S = \{S_1,S_2,\ldots,S_k\}$, there is a constant cost share
    vector $C_S=(c_{S_1},c_{S_2},\ldots,c_{S_k})$ with $c_{S_i}\ge 0$ and
    $\sum_{1\le i\le k} c_{S_i}=1$.  $c_{S_i}$ is agent $S_i$'s cost share
    under coalition $S$. Agents in $S$ unanimously approve the cost share
    vector $C_S$ if and only if $v_{S_i}\ge c_{S_i}$ for all $i$.

    The mechanism picks the largest coalition $S^*$ satisfying that $C_{S^*}$
    is unanimously approved.  If $S^*$ does not exist, then the decision is not
    to build.  If $S^*$ exists, then it is always unique, in which case the
    decision is to build. Only agents in $S^*$ are consumers of the public
    project and they pay according to $C_{S^*}$.

    If agent $i$ belongs to two coalitions $S$ and $T$ with $S\subsetneq T$,
    then $i$'s cost share under $S$ must be greater than or equal to her cost
    share under $T$. Let $N$ be the set of all agents. One way to interpret the
    mechanism is that the agents start with the cost share vector $C_N$.  If
    some agents do not approve their cost shares, then they are forever
    removed.  The remaining agents face new and increased cost shares.  We
    repeat the process until all remaining agents approve their shares, or when
    all agents are removed.  The largest unanimous mechanism is strategy-proof
and individually rational.  \end{definition}

\begin{theorem}[Excludable mech.
    characterization~\cite{Ohseto2000:Characterizations}] For the excludable
    public project model, if a mechanism is strategy-proof, individually
    rational, and satisfies the following assumptions, then it is weakly Pareto
    dominated by a largest unanimous mechanism.

Demand monotonicity: Let $S$ be the set of consumers.  If for every agent $i$
    in $S$, $v_i$ stays the same or increases, then all agents in $S$ are still
    consumers.  If for every agent $i$ in $S$, $v_i$ stays the same or
    increases, and for every agent $i$ not in $S$, $v_i$ stays the same or
    decreases, then the set of consumers should still be $S$.

    Access independence: For all $v_{-i}$, there exist $v_i$ and $v_i'$ so that
agent $i$ is a consumer under type profile $(v_i,v_{-i})$ and is not a consumer
under type profile $(v_i',v_{-i})$.  \end{theorem}

\begin{example}[Serial cost sharing mechanism~\cite{Moulin1994:Serial}] Here is
    an example largest unanimous mechanism.  For every nonempty subset of
    agents $S$ with $|S|=k$, the cost share vector is
    $(\frac{1}{k},\frac{1}{k},\ldots,\frac{1}{k})$.  The mechanism picks the
    largest coalition where the agents are willing to pay equal shares.
\end{example}

%     Another interpretation of this mechanism is the follow iterative process:
%     Initially, all $n$ agents are offered a cost share of $\frac{1}{n}$. If all
%     agents agree, then we build the project and every agent pays $\frac{1}{n}$.
%     If some agents disagree, then they are forever removed.  The remaining $k$
%     agents are offered a cost share of $\frac{1}{k}$. We repeat this until all
%     remaining agents agree or all agents are removed. If all agents are
%     removed, then the decision is not to build.  An agent may face multiple
%     cost share offers throughout the process, but the cost share offers are
%     nondecreasing, so there is no room for strategic manipulation.

    % When an agent reports truthfully under a largest unanimous mechanism, her
    % utility is maximized because her cost share is completely determined by the
    % others. For example, if there are three agents in total, and agents $2$ and
    % $3$ report $\frac{1}{2}$ and $\frac{1}{3}$ respectively, then agent $1$'s
    % cost share offer is $\frac{1}{3}$. Agent $1$ is a consumer if and only if
    % she can afford $\frac{1}{3}$. If agents $2$ and $3$ report $\frac{1}{2}$
    % and $\frac{1}{4}$ respectively, then agent $1$'s cost share offer is
%$\frac{1}{2}$.

Deb and Razzolini~\cite{Deb1999:Voluntary} proved that if we further require an
\emph{equal treatment of equals} property (if two agents have the same type,
then they should be treated the same), then the only strategy-proof and
individually rational mechanism left is the serial cost sharing mechanism.  For
many distributions, we are able to outperform the serial cost sharing mechanism.
That is, equal treatment of equals (or requiring anonymity) may hurt
performances.

\subsection{Nonexcludable Public Project Analysis}\label{sub:nonexcludable}

We start with an analysis on the nonexcludable public project. The results
presented in this section will lay the foundation for the more complex
excludable public project model coming up next.

Due to the characterization results, we focus on the family of unanimous
mechanisms. That is, we are solving for the optimal cost share vector
$(c_1,c_2,\ldots,c_n)$, satisfying that $c_i\ge 0$ and $\sum c_i=1$.

Recall that $f$ and $F$ are the PDF and CDF of the prior distribution.  The
\emph{reliability function} $\overline{F}$ is defined as $\overline{F}(x)=1-F(x)$.  We
define $w(c)$ to be the expected utility of an agent when her cost share is
$c$, conditional on that she accepts this cost share.
\[w(c)=\frac{\int_c^1 (x-c)f(x)dx}{\int_c^1f(x)dx}\]
One condition we will use is \emph{log-concavity}:
if $\log(f(x))$ is concave in $x$, then $f$ is log-concave.
We also introduce another condition called \emph{welfare-concavity}, which requires $w$ to be concave.%, which is equivalent to $f(x)+xf'(x)\ge 0$.

\begin{theorem}\label{thm:nonexcludable}
If $f$ is log-concave, then the conservative equal costs mechanism maximizes the expected
number of consumers.
If $f$ is log-concave and welfare-concave, then the conservative equal costs mechanism
maximizes the expected social welfare.
\end{theorem}

\begin{proof} Let $C=(c_1,c_2,\ldots,c_n)$ be the cost share vector. Maximizing
    the expected number of consumers is equivalent to maximizing the
    probability of $C$ getting unanimously accepted, which equals $\overline{F}(c_1)
    \overline{F}(c_2) \ldots \overline{F}(c_n)$.  Its log equals
    $\sum_{1\le i\le n}\log(\overline{F}(c_i))$.  When $f$ is log-concave, so is
    $\overline{F}$ according to~\cite{Bagnoli2005:Log}.  This means that when cost
    shares are equal, the above probability is maximized.

    The expected social welfare
    under the cost share vector $C$ equals $\sum w(c_i)$, conditional on all
    agents accepting their shares. This is maximized when shares are equal.
    Furthermore, when all shares are equal, the probability of unanimous
approval is also maximized.  \end{proof}

$f$ being log-concave is also called the \emph{decreasing reversed failure
rate} condition~\cite{Shao2016:Optimal}.  Bagnoli and
Bergstrom~\cite{Bagnoli2005:Log} proved log-concavity for many common
distributions, including the distributions in Table~\ref{tb:logconcave} (for
all distribution parameters).  All distributions are restricted to $[0,1]$.
We also list some limited results for welfare-concavity.
We prove that the uniform distribution is welfare-concave, but for the other
distributions, the results are based on simulations.
Finally, we include the conditions for $f$ being nonincreasing, which will be used in the excludable public project model.

\begin{table}[ht]
\caption{Example Log-Concave Distributions}
\centering
\begin{tabular}{ l c r }\label{tb:logconcave}
    & Welfare-Concavity & Nonincreasing \\
    Uniform $U(0,1)$ & Yes & Yes \\
    \hline
    %Normal & $\sigma^2+\mu\ge 1$ & $\mu\le 0$ \\
    Normal & No ($\mu=0.5,\sigma=0.1$) & $\mu\le 0$ \\
    \hline
  % Exponential & $\lambda\le 1$ & Yes \\
    Exponential & Yes ($\lambda=1$) & Yes \\
    \hline
  %Logistic & $1-(e^{-\frac{1-\mu}{s}})^2\le s(1+e^{-\frac{1-\mu}{s}})^2$ & $\mu\le 0$ \\
    Logistic & No ($\mu=0.5,\sigma=0.1$) & $\mu\le 0$ \\
\end{tabular}
\end{table}

% \begin{itemize}
%     \item Uniform $U(0,1)$ is log-concave, nonincreasing, and welfare-concave.
%     \item Normal $N(\mu,\sigma)$ (restricted to $[0,1]$) is log-concave.
%         When $\mu\le 0$, it is nonincreasing.
%         When $\sigma^2+\mu\ge 1$, it is welfare-concave.
%     \item Exponential distribution (restricted to $[0,1]$) is log-concave and nonincreasing.
%         When $\lambda\le 1$, it is welfare-concave.
% \end{itemize}

% \begin{itemize}
%     \item Uniform $U(0,1)$ is both log-concave and welfare-concave.
%     \item Normal $N(\mu,\sigma)$ (restricted to $[0,1]$) is log-concave. When $\sigma^2+\mu\ge 1$, it is welfare-concave.
%     \item Exponential distribution (restricted to $[0,1]$) with any $\lambda$ is log-concave. When $\lambda\le 1$, it is welfare-concave.
%     \item Logistic distribution (restricted to $[0,1]$) with any $\mu$ and $s$ is log-concave. When $\frac{1-(e^{-\frac{1-\mu}{s}})^2}{(1+e^{-\frac{1-\mu}{s}})^2} \le s$, it is welfare-concave.
% \end{itemize}

Even when optimal, the conservative equal costs mechanism performs poorly.  We
take the uniform $U(0,1)$ distribution as an example. Every agent's cost share
is $\frac{1}{n}$.  The probability of acceptance for one agent is
$\frac{n-1}{n}$, which approaches $1$ asymptotically. However, we need
unanimous acceptance, which happens with much lower probability.  For the
uniform distribution, asymptotically, the probability of unanimous acceptance
is only $\frac{1}{e}\approx 0.368$. In general, we have the following bound:

\begin{theorem}
If $f$ is Lipschitz continuous, then when $n$ goes to infinity, the probability of unanimous
    acceptance under the conservative equal costs mechanism is $e^{-f(0)}$.
\end{theorem}

Without log-concavity, the conservative equal costs mechanism is not
necessarily optimal. We present the following dynamic program (DP) for calculating
the optimal unanimous mechanism. We only present the formation for welfare
maximization.\footnote{Maximizing the expected number of consumers can be viewed as a
special case where every agent's utility is $1$ if the project is built}

We assume that there is an ordering of the agents based on their identities.  We define
$B(k,u,m)$ as the maximum expected social welfare under the following conditions:

\begin{itemize}
    \item The first $n-k$ agents have already approved their cost shares, and their total
        cost share is $1-m$. That is, the remaining $k$ agents need to come up with $m$.
    \item The first $n-k$ agents' total expected utility is $u$.
\end{itemize}

The optimal social welfare is then $B(n,0,1)$. We recall that $\overline{F}(c)$ is the probability
that an agent accepts a cost share of $c$, we have
\[
    B(k,u,m)=\max_{0\le c\le m}\overline{F}(c)B(k-1,u+w(c), m-c)
\]
The base case is $B(1,u,m)=\overline{F}(m)(u+w(m))$.  In terms of implementation of
this DP, we have $0\le u\le n$ and $0\le m\le 1$. We
need to discretize these two intervals. If we pick a discretization size of
$\frac{1}{H}$, then the total number of DP subproblems is
about $H^2n^2$.

% We define $P(k,m,b)$ as the maximum probability for $k$ agents to successfully
% cost share $m$ and achieve an expected total welfare of $b$. The optimal social
% welfare is then $\max_b P(n,1,b)b$.

% If $w(m) = \int_m^1 xf(x)dx \ge b$ then $P(1,m,b) = \overline{F}(m)$. When $k\ge 2$, we have

% \[
%     P(k,m,b) = \max_{0\le c\le m} P(k-1,m-c,b-w(c))\overline{F}(c)
% \]

% Here, $m$ and $b$ are both real values between $0$ and $1$. In our computation,
% we discretize $[0,1]$ into grid points
% $\{0,\frac{1}{H},\frac{2}{H},\ldots,1\}$. The number of DP subproblems is then
% $n(H+1)^2$.

To compare the performance of the conservative equal costs mechanism and our DP
solution, we focus on distributions that are not log-concave (hence, uniform
and normal are not eligible).  We introduce the following non-log-concave
distribution family:

\begin{definition}[Two-Peak Distribution $(\mu_1,\sigma_1,\mu_2,\sigma_2,p)$]
    With probability $p$, the agent's valuation is drawn from the normal
    distribution $N(\mu_1,\sigma_1)$ (restricted to $[0,1]$).
    With probability $1-p$, the agent's valuation is drawn from $N(\mu_2,\sigma_2)$ (restricted to $[0,1]$).
\end{definition}

The motivation behind the two-peak distribution is that there may be two
categories of agents. One category is directly benefiting from the public
project, and the other is indirectly benefiting. For example, if the public
project is to build bike lanes, then cyclists are directly benefiting, and the
other road users are indirectly benefiting (\emph{e.g.}, less congestion for
them).  As another example, if the public project is to crowdfund a piece of
security information on a specific software product (\emph{e.g.}, PostgreSQL),
then agents who use PostgreSQL in production are directly benefiting and the
other agents are indirectly benefiting (\emph{e.g.}, every web user is pretty
much using some websites backed by PostgreSQL).  Therefore, it is natural to
assume the agents' valuations are drawn from two different
distributions. For simplicity, we do not consider three-peak, \emph{etc.}

For the two-peak distribution $(0.1,0.1,0.9,0.1,0.5)$, DP significantly
outperforms the conservative equal costs (CEC) mechanism.

\begin{center}
\begin{tabular}{ l c r }
    & E(no. of consumers) & E(welfare)\\
  n=3 CEC & 0.376 & 0.200 \\
    \hline
  n=3 DP & 0.766 & 0.306 \\
    \hline
  n=5 CEC & 0.373 & 0.199 \\
    \hline
  n=5 DP & 1.426 & 0.591 \\
\end{tabular}
\end{center}

\subsection{Excludable Public Project}
Due to the characterization results, we focus on the family of largest
unanimous mechanisms.  We start by showing that the serial cost sharing
mechanism is optimal in some scenarios.

\begin{theorem}\label{thm:excludable}\,
$2$ agents case:
If $f$ is log-concave, then the serial cost sharing mechanism maximizes the expected
    number of consumers.
If $f$ is log-concave and welfare-concave, then the serial cost sharing mechanism
maximizes the expected social welfare.

$3$ agents case:
If $f$ is log-concave and nonincreasing, then the serial cost sharing mechanism maximizes the expected
    number of consumers.
If $f$ is log-concave, nonincreasing, and welfare-concave, then the serial cost sharing mechanism maximizes the social welfare.
\end{theorem}

For $2$ agents, the conditions are identical to the nonexcludable case.  For
$3$ agents, we also need $f$ to be nonincreasing.  Example distributions
satisfying these conditions were listed in Table~\ref{tb:logconcave}.

\begin{proof}
We only present the proof for welfare maximization when $n=3$, which is the most complex case.
    (For maximizing the number of consumers, all references to the $w$ function should be
    replaced by the constant $1$.)
The largest unanimous mechanism specifies constant cost
shares for every coalition of agents.
We use $c_{1\underline{2}3}$ to denote agent $2$'s cost share when the coalition is $\{1,2,3\}$.
Similarly, $c_{\underline{2}3}$ denotes agent $2$'s cost share when the coalition is $\{2,3\}$.
If the largest unanimous coalition has size $3$, then the expected social welfare gained due to this
    case is:
    \[
        \overline{F}(c_{\underline{1}23})
        \overline{F}(c_{1\underline{2}3})
        \overline{F}(c_{12\underline{3}})
        (
        w(c_{\underline{1}23})
        +w(c_{1\underline{2}3})
        +w(c_{12\underline{3}})
        )
    \]
    Given log-concavity of $\overline{F}$ (implied by the log-concavity of $f$) and welfare-concavity,
    and given that $c_{\underline{1}23}+c_{1\underline{2}3}+c_{12\underline{3}}=1$. We have that the above is maximized when all agents have equal shares.

    If the largest unanimous coalition has size $2$ and is $\{1,2\}$, then the expected social welfare gained due to this
    case is:
    \[
        \overline{F}(c_{\underline{1}2})
        \overline{F}(c_{1\underline{2}})
        F(c_{12\underline{3}})
        (
        w(c_{\underline{1}2})
        +w(c_{1\underline{2}})
        )
    \]
    $F(c_{12\underline{3}})$ is the probability that agent $3$ does not join in the coalition.
    The above is maximized when
    $c_{\underline{1}2}=c_{1\underline{2}}$, so it simplifies to
    $2\overline{F}(\frac{1}{2})^2 w(\frac{1}{2}) F(c_{12\underline{3}})$.
    We then consider the welfare gain from all coalitions of size $2$:
    \[
        2\overline{F}(\frac{1}{2})^2
        w(\frac{1}{2})(
        F(c_{\underline{1}23})
        +F(c_{1\underline{2}3})
        +F(c_{12\underline{3}})
        )
    \]
    Since $f$ is nonincreasing, we have that $F$ is concave, the above is again maximized when all cost shares are equal.

    Finally, the probability of coalition size $1$ is $0$, which can be ignored in our analysis.
    Therefore, throughout the proof, all terms referenced are maximized when the cost shares are equal.
\end{proof}

For $4$ agents and uniform distribution, we have a similar result.

\begin{theorem}\label{thm:uniform}
    Under the uniform distribution $U(0,1)$, when $n=4$, the serial cost sharing
    mechanism maximizes the expected number of consumers and the expected social welfare.
\end{theorem}

For $n\ge 4$ and for general distributions, we propose a numerical method for
calculating the performance upper bound.  A largest unanimous mechanism can be
carried out by the following process: we make cost share offers to the agents
one by one based on an ordering of the agents. Whenever an agent disagrees, we
remove this agent and move on to a coalition with one less agent. We repeat
until all agents are removed or all agents have agreed. We introduce the
following mechanism based on a Markov process.  The initial state is
$\{(\underbrace{0,0,\ldots,0}_n),n\}$, which represents that initially, we only
know that the agents' valuations are at least $0$, and we have not made any
cost share offers to any agents yet (there are $n$ agents yet to be offered).
We make a cost share offer $c_1$ to agent $1$.  If agent $1$ accepts, then we
move on to state $\{(c_1,\underbrace{0,\ldots,0}_{n-1}),n-1\}$. If agent $1$
rejects, then we remove agent $1$ and move on to reduced-sized state
$\{(\underbrace{0,\ldots,0}_{n-1}),n-1\}$. In general, let us consider a state
with $t$ users $\{(l_1,l_2,\ldots,l_t),t\}$. The $i$-th agent's valuation lower
bound is $l_i$. Suppose we make offers $c_1,c_2,\ldots,c_{t-k}$ to the first
$t-k$ agents and they all accept, then we are in a state
$\{(\underbrace{c_1,\ldots,c_{t-k}}_{t-k},\underbrace{l_{t-k+1},\ldots,l_{t}}_k),k\}$.
The next offer is $c_{t-k+1}$. If the next agent accepts, then we move on to
$\{(\underbrace{c_1,\ldots,c_{t-k+1}}_{t-k+1},\underbrace{l_{t-k+2},\ldots,l_{t}}_{k-1}),k-1\}$.
If she disagrees (she is then the first agent to disagree), then we move on to
a reduced-sized state
$\{(\underbrace{c_1,\ldots,c_{t-k}}_{t-k},\underbrace{l_{t-k+2},\ldots,l_{t}}_{k-1}),t-1\}$.
Notice that whenever we move to a reduced-sized state, the number of agents yet
to be offered should be reset to the total number of agents in this state.
Whenever we are in a state with all agents offered
$\{(\underbrace{c_1,\ldots,c_t}_t),0\}$, we have gained an objective value of
$t$ if the goal is to maximize the number of consumers.  If the goal is to
maximize welfare, then we have gained an objective value of $\sum_{1\le i\le t}
w(c_i)$.  Any largest unanimous mechanism can be represented via the above
Markov process.  So for deriving performance upper bounds, it suffices to focus
on this Markov process.

% When calculating the performance upper bound, we make one relaxation. That
% is, for two states $(c_1,c_2,\ldots,c_k)$ and $(c_1',c_2',\ldots,c_k')$, with
% $\sum c_i=\sum c_i'=c_s$, we are able to jump between them, and we always
% jump to the best state.  Essentially, we reduce the state to a one
% dimensional value $c_s$.

Starting from a state, we may end up with different objective values. A state
has an expected objective value, based on all the transition probabilities.  We
define $U(t,k,m,l)$ as the maximum expected objective value starting from a
state that satisfies:

\begin{itemize}
    \item There are $t$ agents in the state.

    \item There are $k$ agents yet to be offered.  The first $t-k$ agents
        (those who accepted the offers) have a total cost share of $1-m$. That
        is, the remaining $k$ agents are responsible for a total cost share of $m$.

    \item The $k$ agents yet to be offered have a total lower bound of $l$.
\end{itemize}

The upper bound we are looking for is then $U(n,n,1,0)$, which can be calculated
via the following DP process:
\[
    U(t,k,m,l) = \max_{\substack{0\le l^*\le l\\l^*\le c^*\le m}} \left( \frac{\overline{F}(c^*)}{\overline{F}(l^*)}U(t,k-1,m-c^*,l-l^*)\right.
\]
\[
    \left.+(1-\frac{\overline{F}(c^*)}{\overline{F}(l^*)})U(t-1,t-1,1,1-m+l-l^*)\right)
\]

In the above, there are $k$ agents yet to be offered. We maximize over the next agent's
possible lower bound $l^*$ and the cost share $c^*$. That is, we look for the
best possible lower bound situation and the corresponding optimal offer.
$\frac{\overline{F}(c^*)}{\overline{F}(l^*)}$ is the probability that the next agent
accepts the cost share, in which case, we have $k-1$ agents left. The remaining
agents need to come up with $m-c^*$, and their lower bounds sum up to $l-l^*$.
When the next agent does not accept the cost share, we transition to a new
state with $t-1$ agents in total. All agents are yet to be offered, so $t-1$
agents need to come up with $1$. The lower bounds sum up to $1-m+l-l^*$.

There are two base conditions.  When there is only one agent, she
has $0$ probability for accepting an offer of $1$, so $U(1,k,m,l) = 0$.
The other base case is that when there is only $1$ agent yet to be offered,
the only valid lower bound is $l$ and the only sensible offer is $m$. Therefore,
\[U(t,1,m,l) = \frac{\overline{F}(m)}{\overline{F}(l)}G(t)+(1-\frac{\overline{F}(m)}{\overline{F}(l)})U(t-1,t-1,1,1-m)\]

Here, $G(t)$ is the maximum
objective value when the largest unanimous set has size $t$.  For maximizing
the number of consumers, $G(t)=t$.  For maximizing welfare,
\[G(t)= \max_{\substack{c_1,c_2,\ldots,c_t\\c_i\ge 0\\\sum c_i=1}}\sum_i w(c_i)\]
The above $G(t)$ can be calculated via a trivial DP. % We ignore the details here.
% Let $G'(k,m)$ be defined as
% \[
%     G'(k,m)= \max_{c_1,c_2,\ldots,c_k, c_i\ge 0, \sum c_i=m}\sum_i w(c_i)
% \]

% We have that $G(k)=G'(k,1)$. When calculating $G'$, the base case is $G'(1,m)=w(m)$, and we have

% \[
%     G'(k,m)= \max_{0\le c\le m} \left(w(c)+G'(k-1,m-c)\right)
% \]

Now we compare the performances of the serial cost sharing mechanism against
the upper bounds.  All distributions used here are log-concave.  In every cell,
the first number is the objective value under serial cost sharing, and the
second is the upper bound.  We see that the serial cost sharing mechanism is
close to optimality in all these experiments.  We include both welfare-concave
and non-welfare-concave distributions (uniform and exponential with $\lambda=1$
are welfare-concave). For the two distributions not satisfying
welfare-concavity, the welfare performance is relatively worse.

\begin{center}
\begin{tabular}{ l c r }
    & E(no. of consumers) & E(welfare)\\
    n=5 $U(0,1)$ & 3.559, 3.753 & 1.350, 1.417 \\
    \hline
    n=10 $U(0,1)$ & 8.915, 8.994& 3.938, 4.037 \\
    \hline
    n=5 $N(0.5,0.1)$ & 4.988, 4.993 & 1.492, 2.017 \\
    \hline
    n=10 $N(0.5,0.1)$ & 10.00, 10.00 & 3.983, 4.545 \\
    \hline
    n=5 Exponential $\lambda=1$ & 2.799, 3.038 & 0.889, 0.928 \\
    \hline
    n=10 Exponential $\lambda=1$ & 8.184, 8.476 & 3.081, 3.163 \\
    \hline
    n=5 Logistic$(0.5,0.1)$ & 4.744, 4.781 & 1.451, 1.910 \\
    \hline
    n=10 Logistic$(0.5,0.1)$ & 9.873, 9.886 & 3.957, 4.487 \\
\end{tabular}
\end{center}

% In summary, for the nonexcludable public project model, if the distribution
% satisfies the condition in Theorem~\ref{thm:nonexcludable}, then the
% conservative cost sharing mechanism is optimal. For other distributions, we
% proposed a dynamic programming technique, which produces the optimal mechanism.
% For the excludable public project model, if the distribution and the number of
% agents satisfy the condition in Theorem~\ref{thm:excludable} or
% Theorem~\ref{thm:uniform}, then the serial cost sharing mechanism is optimal.
% For a few common distributions, we calculated the performance upper bound, and
% demonstrated that the serial cost sharing mechanism is close to optimality.
% However, there does exist distribution under which the serial cost sharing is
% far away from optimality.

\begin{example} Here we provide an example to show that the serial cost sharing
    mechanism can be far away from optimality. We pick a simple Bernoulli
    distribution, where an agent's valuation is $0$ with $0.5$ probability and
    $1$ with $0.5$ probability.\footnote{Our paper assumes that the
    distribution is continuous, so technically we should be considering a
    smoothed version of the Bernoulli distribution. For the purpose of
    demonstrating an elegant example, we ignore this technicality.} Under the
    serial cost sharing mechanism, when there are $n$ agents, only half of the
    agents are consumers (those who report $1$s). So in expectation, the number
    of consumers is $\frac{n}{2}$.  Let us consider another simple mechanism.
    We assume that there is an ordering of the agents based on their identities
    (not based on their types). The mechanism asks the first agent to accept a
    cost share of $1$. If this agent disagrees, she is removed from the system.
    The mechanism then moves on to the next agent and asks the same, until an
    agent agrees. If an agent agrees, then all future agents can consume the
    project for free. The number of removed agents follows a geometric
    distribution with $0.5$ success probability. So in expectation, $2$
    agents are removed.  That is, the expected number of consumers is $n-2$.
\end{example}

% For welfare maximization, we do not have an elegant example, but our results in the next section
% show that there exists better mechanism for some distributions.

\section{Mech. Design vs Neural Networks}
% We recall that for the nonexcludable public project problem, we already have
% a DP algorithm for calculating the optimal mechanism for general
% distributions.  For the excludable public project model, experiments suggest
% that with log-concavity, the serial cost sharing mechanism is nearly optimal.
% For the
For the rest of this paper, we focus on the excludable public project model and
distributions that are not log-concave.  Due to the characterization results,
we only need to consider the largest unanimous mechanisms. We use neural
networks and deep learning to solve for well-performing largest unanimous
mechanisms. Our approach involves several technical innovations as discussed in
Section~\ref{sec:intro}.

\subsection{Network Structure} A largest unanimous mechanism specifies constant
cost shares for every coalition of agents. The mechanism can be characterized
by a neural network with $n$ binary inputs and $n$ outputs. The $n$ binary
inputs present the coalition, and the $n$ outputs represent the constant cost
shares.  We use $\vec{b}$ to denote the input vector (tensor) and $\vec{c}$ to
denote the output vector. We use $NN$ to denote the neural network, so
$NN(\vec{b})=\vec{c}$.

There are several constraints on the neural network.

\begin{itemize}
    \item All cost shares are nonnegative: $\vec{c}\ge 0$.

    \item For input coordinates that are $1$s, the output coordinates should
        sum up to $1$.  For example, if $n=3$ and $\vec{b}=(1,0,1)$ (the
        coalition is $\{1,3\}$), then $\vec{c}_1+\vec{c}_3=1$ (agent $1$ and
        $3$ are to share the total cost).

    \item For input coordinates that are $0$s, the output coordinates are
        irrelevant. We set these output coordinates to $1$s, which makes it
        more convenient for the next constraint.

    \item Every output coordinate is nondecreasing in every input coordinate.
        This is to ensure that the agents' cost shares are nondecreasing when
        some other agents are removed. If an agent is removed, then her cost
        share offer is kept at $1$, which makes it trivially nondecreasing.
\end{itemize}

All constraints except for the last is easy to achieve.
We will simply use $OUT(\vec{b})$ as output instead of directly using $NN(\vec{b})$\footnote{This is done by appending additional calculation structures to the output layer.}:
\[OUT(\vec{b})=\text{softmax}(NN(\vec{b})-1000(1-\vec{b}))+(1-\vec{b})\]

Here, $1000$ is an arbitrary large constant.
For example, let $\vec{b}=(1,0,1)$ and $\vec{c}=NN(\vec{b})=(x,y,z)$. We have
\[OUT(\vec{b})=\text{softmax}((x,y,z)-1000(0,1,0))+(0,1,0)\]
\[=\text{softmax}((x,y-1000,z))+(0,1,0)\]
\[=(x',0,z')+(0,1,0)=(x',1,y')\]

In the above, $\text{softmax}((x,y-1000,z))$ becomes $(x',0,y')$ with $x',y'\ge
0$ and $x'+y'=1$ because the second coordinate is very small so it
(essentially) vanishes after softmax. Softmax always produces nonnegtive
outputs that sum up to $1$.  Finally, the $0$s in the output are flipped to
$1$s per our third constraint.

The last constraint is enforced using a penalty function.
For $\vec{b}$ and $\vec{b}'$, where $\vec{b}'$ is obtained from $\vec{b}$ by changing one $1$ to $0$,
we should have that $OUT(\vec{b})\le OUT(\vec{b}')$, which leads to the following penalty (times a large constant):
\[\text{ReLU}(OUT(\vec{b})-OUT(\vec{b}'))\]

Another way to enforce the last constraint is to adopt the idea behind
Sill~\cite{Sill1998:Monotonic}.  The authors proposed a network structure
called the \emph{monotonic networks}. This idea has been used
in~\cite{Golowich2018:Deep}, where the authors also dealt with networks that
take binary inputs and must be monotone.  However, we do not use this approach
because it is incompatible with our design for achieving the other constraints.
There are two other reasons for not using the monotonic network structure. One
is that it has only two layers. Some argue that having a \emph{deep} model is
important for performance in deep learning~\cite{Zhou2017:Deep}.  The other is
that under our approach, we only need a fully connected network with ReLU
penalty, which is highly optimized in state-of-the-art deep learning toolsets.
In our experiments, we use a fully connected network with four layers ($100$
nodes each layer) to represent our mechanism.

\subsection{Cost Function}
For presentation purposes, we focus on maximizing the expected number of
consumers.  Only slight adjustments are needed for welfare maximization.

Previous approaches of mechanism design via neural networks used \emph{static}
networks~\cite{Golowich2018:Deep,
Duetting2019:Optimal,Shen2019:Automated,Manisha2018:Learning}. Given a sample,
the mechanism simulation is done on the network.  Our largest unanimous
mechanism involves iterative decision making.
% We recall that the largest unanimous mechanism can be interpreted as follows.
% We start with all agents.  Only agents who approve their cost shares remain
% in the system. The other agents are forever removed. The remaining agents
% form a new coalition, which leads to a new cost share vector.  We repeat
% until all agents agree or all agents are removed.
We actually can model the process via a static network, but the result is not
good.  The initial offers are $OUT((1,1,\ldots,1))$.  The remaining agents
after the first round are then $S=\text{sigmoid}(v-OUT((1,1,\ldots,1)))$.
Here, $v$ is the type profile sample. The sigmoid function turns positive
values to (approximately) $1$s and negative values to (approximately)  $0$s.
The next round of offers are then $OUT(S)$. The remaining agents afterwards are
then $\text{sigmoid}(v-OUT(S))$.  We repeat this $n$ times because the largest
unanimous mechanism must terminate after $n$ rounds.  The final coalition is a
converged state, so even if the mechanism terminates before the $n$-th round,
having it repeat $n$ times does not change the result (except for additional
numerical errors).  Once we have the final coalition $S^f$, we include
$\sum_{x\in S^f}x$ (number of consumers) in the cost function.\footnote{Has to multiply
$-1$ as we typically minimize the cost function.} However, this approach
performs \emph{abysmally}, possibly due to the vanishing gradient problem and
numerical errors caused by stacking $n$ sigmoid functions.

We would like to avoid stacking sigmoid to model iterative decision making (or
get rid of sigmoid altogether). We propose an alternative approach, where
decisions are simulated off the network using a separate program (\emph{e.g.},
any Python function). The advantage of this approach is that it is now trivial
to handle complex decision making.  However, experienced neural network
practitioners may immediately notice a pitfall.  Given a type profile sample
$v$ and the current network $NN$, if we simulate the mechanism off the network
to obtain the number of consumers $x$, and include $x$ in the cost function,
then training will fail completely. This is because $x$ is a constant that
carries no gradients at all.\footnote{We use PyTorch in our experiments. An
overview of Automated Differentiation in PyTorch is available here~\cite{Paszke2017:Automatic}.}

One way to resolve this is to interpret the mechanisms as price-oriented
rationing-free (PORF) mechanisms~\cite{Yokoo2003:Characterization}.  That is,
if we single out one agent, then her options (outcomes combined with payments)
are completely determined by the other agents and she simply has to choose the
utility-maximizing option.  Under a largest unanimous mechanism, an agent faces
only two results: either she belongs to the largest unanimous coalition or not.
If an agent is a consumer, then her payment is a constant due to
strategy-proofness, and the constant payment is determined by the other agents.
Instead of sampling over complete type profiles, we sample over $v_{-i}$ with a
random $i$.  To better convey our idea, we consider a specific example.  Let
$i=1$ and $v_{-1}=(\cdot, \frac{1}{2},\frac{1}{2},\frac{1}{4}, 0)$.  We assume
that the current state of the neural network is exactly the serial cost sharing
mechanism.  Given a sample, we use a separate program to calculate the
following entries. In our experiments, we simply used Python simulation to
obtain these entries.

\begin{itemize}

    \item The objective value if $i$ is a consumer ($O_s$). Under the example,
        if $1$ is a consumer, then the decision must be $4$ agents each pays
        $\frac{1}{4}$. So the objective value is $O_s=4$.

    \item The objective value if $i$ is not a consumer ($O_f$). Under the
        example, if $1$ is not a consumer, then the decision must be $2$ agents
        each pay $\frac{1}{2}$. So the objective value is $O_f=2$.

    \item The binary vector that characterizes the coalition that decides $i$'s
offer ($\vec{O_b}$). Under the example, the vector is $\vec{O_b}=(1,1,1,1,0)$.  \end{itemize}

$O_s$, $O_f$, and $\vec{O_b}$ are constants without gradients. We link them together
using terms with gradients, which is then included in the cost function:
\begin{equation}\label{eq:single1}
    (1-F(OUT(\vec{O_b})_i))O_s + F(OUT(\vec{O_b})_i)O_f
\end{equation}

$1-F(OUT(\vec{O_b})_i)$ is the probability that agent $i$ accepts her offer.
$F(OUT(\vec{O_b})_i)$ is then the probability that agent $i$ rejects her offer.
$OUT(\vec{O_b})_i$ carries gradients as it is generated by the network.  We use the
analytical form of $F$, so the above term carries gradients.\footnote{PyTorch
has built-in analytical CDFs of many common distributions.}

The above approach essentially feeds the prior distribution into the cost
function. We also experimented with two other approaches. One does not use the
prior distribution. It uses a full profile sample and uses one layer of sigmoid to select
between $O_s$ or $O_f$:
\begin{equation}\label{eq:sigmoid}
    \text{sigmoid}(v_i-OUT(\vec{O_b})_i)O_s + \text{sigmoid}(OUT(\vec{O_b})_i-v_i))O_f
\end{equation}

The other approach is to feed ``even more'' distribution into the cost
function.  We single out two agents $i$ and $j$. Now there are $4$ options:
they both win or both lose, only $i$ wins, and only $j$ wins. We still use $F$
to connect these options together.

In Section~\ref{sec:experiment}, in one experiment, we show that singling out
one agent works the best. In another experiment, we show that even if we do not
have the analytical form of $F$, using an analytical approximation also enables
successful training.

\subsection{Supervision as Initialization} We introduce an additional
supervision step in the beginning of the training process as a systematic way
of initialization.  We first train the neural network to mimic an existing
manual mechanism, and then leave it to gradient descent.  We considered three
different manual mechanisms. One is the serial cost sharing mechanism.  The
other two are based on two different heuristics:

% If we already have an existing manual mechanism, but we
% like to see whether we can improve it via computational methods.  Then our
% approach can be applied. We first learn toward the manual mechanism, then with
% gradient descent, we are able to achieve better performances.  Another
% application is that if we have a manual heuristic-based mechanism, but due to
% its heuristic nature, sometimes it violate the mechanism constraints, then we
% could first teach the network this heuristic, and then use gradient descent to
% fix the constraint violations. Finally, having a good initial state may greatly
% reduce the training speed. Our examples show that for some scenarios, it takes
% a long time for the random initialized approach (without the supervision step)
% to catch up. For computationally expensive experiments, sometimes it never
% catches up.\footnote{We typically run our algorithm for a few hours. If may
% take a few days for the random initialization to catch up, so for some
% experiments, within our time scope, the random initialization version
% effectively never caught up.}

\begin{definition}[One Directional Dynamic Programming] We make offers to the
    agents one by one. Every agent faces only one offer. The offer is based on
    how many agents are left, the objective value cumulated so far by the
    previous agents, and how much money still needs to be raised.  If an agent
    rejects an offer, then she is removed from the system.  At the end of the
    algorithm, we check whether we have collected $1$.  If so, the project is
    built and all agents not removed are consumers.  This mechanism belongs to
the largest unanimous mechanism family. This mechanism is not optimal because
we cannot go back and increase an agent's offer. \end{definition}

\begin{definition}[Myopic Mechanism] For coalition size $k$, we treat it as a
    nonexcludable public project problem with $k$ agents.  The offers are
    calculated based on the dynamic program proposed at the end of
    Subsection~\ref{sub:nonexcludable}, which computes the optimal offers for
    the nonexcludable model.  This is called the myopic mechanism, because it
    does not care about the payoffs generated in future rounds.  This mechanism
    is not necessarily feasible, because the agents' offers are not necessarily
nondecreasing when some agents are removed.  \end{definition}

% For the above mechanism, the offer is calculated as follows: We introduce
% $DP(k,g,m)$ as the expected number of consumers when the last $k$ agents are
% responsible for a total cost share of $m$, and $g$ among the first $n-k$ agents
% have agreed to their offers. $DP(n,0,1)$ is then the overall objective.
% \[DP(k,g,m)=\max_{0\le c\le m}\bar{F}(c)DP(k-1,g+1,m-c)+F(c)DP(k-1,g,m)\] The
% optimal offer for the next agent (among the remaining $k$ agents) is then the
% optimal $c$ in the above. The base case is $DP(1,g,m)=\bar{F}(m)(g+1)$, since
% if $m$ is left for the last agent, then her share must be $m$.

\section{Experiments}\label{sec:experiment} The experiments are conducted on a
machine with Intel i5-8300H CPU.\footnote{We experimented with both PyTorch and
Tensorflow (eager mode).  The PyTorch version runs significantly faster,
possibly because we are dealing with dynamic graphs.} The largest experiment
with $10$ agents takes about $3$ hours. Smaller scale experiments take only
about $15$ minutes.
% \footnote{As a comparison, \cite{Duetting2019:Optimal}
% reports that it takes $9$ hours to solve a $2$ bidders $3$ items case, on a
% cluster of GPUs. Of course, the mechanism design setting is different and our
% approach is not as general as the one considered
% in~\cite{Duetting2019:Optimal}.}

In our experiments, unless otherwise specified, the distribution considered is
two-peak $(0.15,0.1,0.85,0.1,0.5)$.  The x-axis shows the number of training
rounds. Each round involves $5$ batches of $128$ samples ($640$ samples each
round).  Unless otherwise specified, the y-axis shows the expected number of
\textbf{non}consumers (so lower values represent better performances).  Random
initializations are based on Xavier normal with bias $0.1$.

Figure~\ref{fig:1} (Left) shows the performance comparison of three different
ways for constructing the cost function: using one layer of sigmoid
(without using distribution) based on~\eqref{eq:sigmoid}, singling out one agent based on~\eqref{eq:single1}, and singling out two agents.
All trials start from random initializations.  In this experiment, singling out
one agent works the best. The sigmoid-based approach is capable of moving the
parameters, but its result is noticeably worse. Singling out two agents has
almost identical performance to singling out one agent, but it is slower in
terms of time per training step.

Figure~\ref{fig:1} (Right) considers the Beta $(0.1,0.1)$ distribution.  We use
Kumaraswamy $(0.1,0.354)$'s analytical CDF to approximate the CDF of Beta
$(0.1,0.1)$. The experiments show that if we start from random initializations
(Random) or start by supervision to serial cost sharing (SCS), then the cost
function gets stuck. Supervision to one directional dynamic programming (DP)
and Myoptic mechanism (Myopic) leads to better mechanisms. So in this example
scenario, approximating CDF is useful when analytical CDF is not available. It
also shows that supervision to manual mechanisms works better than random
initializations in this case.

%\vspace{-.2in}
\begin{figure}%[H]
    \caption{Effect of Distribution Info on Training}
\centering
\includegraphics[width=\textwidth]{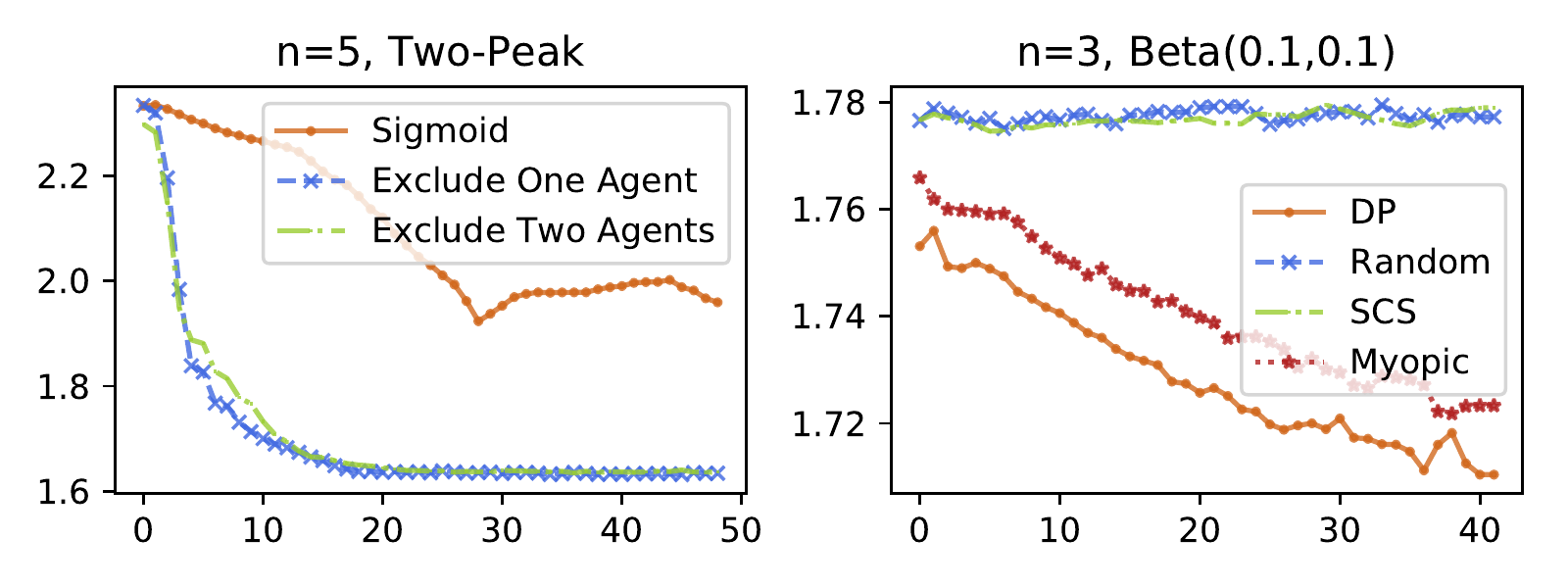}\label{fig:1}
\end{figure}
%\vspace{-.2in}

Figure~\ref{fig:2} (Top-Left $n=3$, Top-Right $n=5$, Bottom-Left $n=10$) shows
the performance comparison of supervision to different manual mechanisms.  For
$n=3$, supervision to DP performs the best. Random initializations is able to
catch up but not completely close the gap. For $n=5$, random initializations
caught up and actually became the best performing one. The Myopic curve first
increases and then decreases because it needs to first fix the constraint
violations. For $n=10$, supervision to DP significantly outperforms the others. Random
initializations closes the gap with regard to serial cost sharing, but it then
gets stuck. Even though it looks like the DP curve is flat, it is actually improving, albeit very slowly. A magnified version is shown in Figure~\ref{fig:2} (Bottom-Right).

%\vspace{-.2in}
\begin{figure}%[H]
    \caption{Supervision to Different Manual Mechanisms}
\centering
    \includegraphics[width=\textwidth]{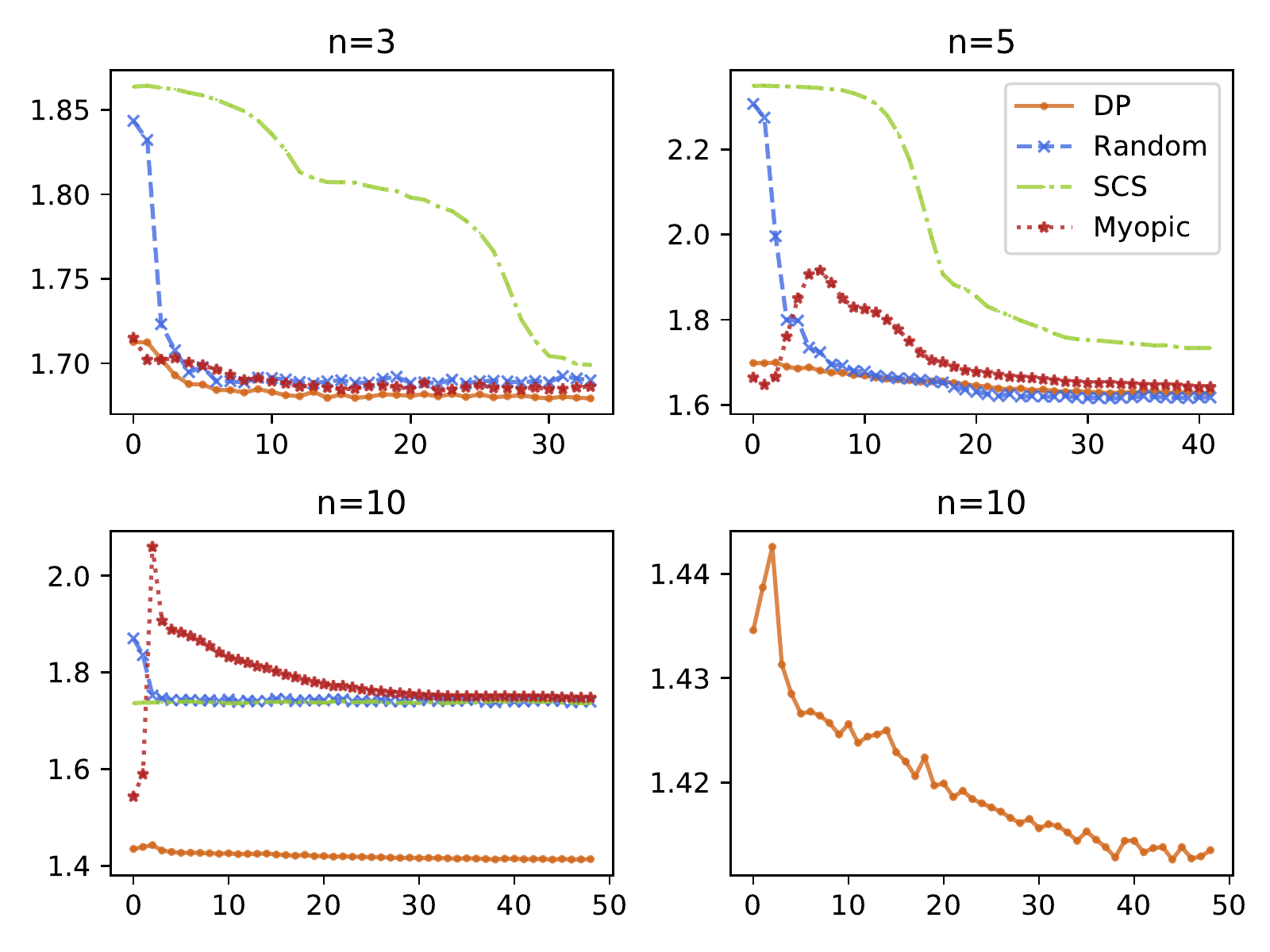}\label{fig:2}
\end{figure}
%\vspace{-.2in}

%We conclude with two experiments on maximizing social welfare.
%Unfortunately, our networks were not able to find noticeable improvement for
%welfare maximization.
Figure~\ref{fig:3} shows two experiments on maximizing expected social welfare (y-axis) under
two-peak $(0.2,0.1,0.6,0.1,0.5)$.  For $n=3$, supervision to DP leads to the
best result. For $n=5$, SCS is actually the best mechanism we can find (the
cost function barely moves).
% In this scenario, all manual mechanisms have very similar welfares: $0.4566$
% (DP), $0.4419$ (SCS), $0.4461$ (Myopic). Even random initialization
% \emph{without training} has a welfare of $0.4428$.  That's at most $5\%$
% difference across these starting points. Our neural networks were That is,
% for $n=5$, the neural networks approach does not lead to better mechanisms.
It should be noted that all manual mechanisms \emph{before training} have very similar welfares:
$0.7517$ (DP), $0.7897$ (SCS), $0.7719$ (Myopic). Even random initialization
before training has a welfare of $0.7648$.  It could be that there is just little room
for improvement here.
%That's at most $5\%$ difference
%across these starting points.
% including random initializations (less than $5\%$ difference). It could be that
% there is actually no room for improvement at all.

%\vspace{-.2in}
\begin{figure}%[H]
    \caption{Maximizing Social Welfare}
\centering
    \includegraphics[width=\textwidth]{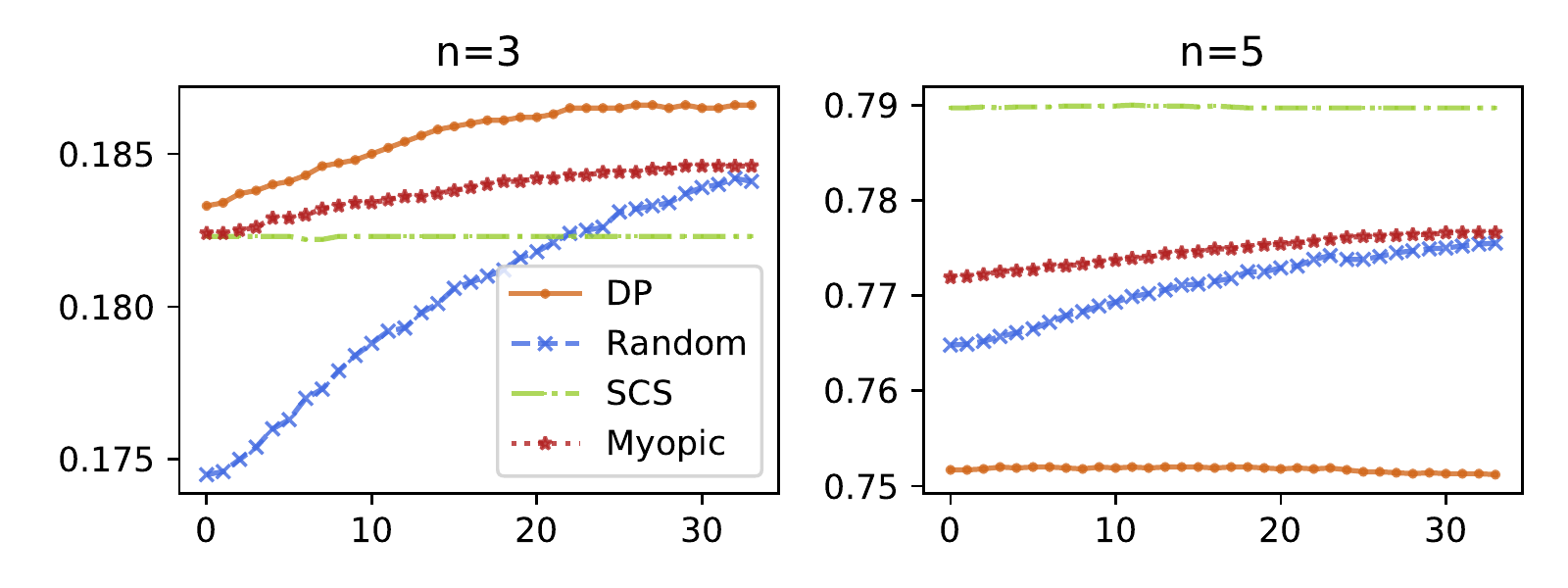}\label{fig:3}
\end{figure}
%\vspace{-1in}

%\bibliographystyle{ACM-Reference-Format}
\bibliographystyle{plain}
\bibliography{/home/mingyu/nixos/newmg.bib}
\end{document}